\documentclass[11pt]{amsart}
\usepackage{geometry}                
\usepackage{graphicx}
\usepackage{amssymb, amsmath, amsthm}
\usepackage{natbib}
\usepackage{epstopdf}
\usepackage{enumerate}
\usepackage{textcomp}
\usepackage{mathrsfs }
\usepackage[section]{algorithm}
\usepackage{algorithmic}
\usepackage{tikz}
\usepackage{pgfplots}
\usepgfplotslibrary{patchplots}
\usepackage{subfigure}
\usepackage{comment}
\usepackage{multirow}

\newtheorem{theorem}{Theorem}[section]
\newtheorem{lemma}[theorem]{Lemma}
\newtheorem{proposition}[theorem]{Proposition}

\theoremstyle{definition}

\newtheorem{remark}[theorem]{Remark}

\newcommand{\F}{\operatorname{\mathcal{F}}}
\newcommand{\iF}{\operatorname{\mathcal{F}^{-1}}}

\newcommand{\Z}{\mathbf{Z}}

\newcommand{\ErrorFrequency}{\mathcal{E_F}}
\newcommand{\ErrorQuadrature}{\mathcal{E_Q}}
\newcommand{\Error}{\mathcal{E}}
\newcommand{\EF}{\ErrorFrequency}
\newcommand{\EQ}{\ErrorQuadrature}
\newcommand{\E}{\Error}
\newcommand{\est}{\overline {\mathcal{E}}}
\newcommand{\estFrequency}{\est_F}
\newcommand{\estQuadrature}{\est_Q}

\newcommand{\lx}{\mathcal{L}^{X}}

\newcommand{\R}{\mathbb{R}}
\newcommand{\C}{\mathbb{C}}

\newcommand{\damp}{\alpha}

\newcommand{\cgmyY}{Y}
\newcommand{\cgmyM}{\eta_+}
\newcommand{\cgmyG}{\eta_-}
\newcommand{\cgmyC}{K}

\newcommand{\fa}{f_{\damp}}
\newcommand{\fadisc}[1]{f_{\damp, #1}}
\newcommand{\fdisc}[1]{f_{#1}}
\newcommand{\ga}{g_{\damp}}

\newcommand{\ft}[1]{\hat{#1}}
\newcommand{\kft}[1]{\tilde{#1}}
\newcommand{\ift}[1]{\iF\left[#1\right]}

\newcommand{\real}[1]{\mathrm{Re}\left[#1\right]}
\newcommand{\im}[1]{\mathrm{Im}\left[#1\right]}
\newcommand{\step}{\Delta\omega}
\newcommand{\strip}[1]{A_{#1}}

\newcommand{\charExp}[1]{\Psi \parent{{#1}}}
\newcommand{\charFun}[2]{\varphi_{#1} \parent{{#2}}}

\newcommand{\cutoff}{\omega_{\mathrm{\max}}}

\newcommand{\pprob}[1]{}

\newcommand{\expf}[1]{\mathrm{exp}\left ( {#1}\right )}
\newcommand{\expfs}[1]{e^{#1}}

\newcommand{\parent}[1]{\left( {#1} \right)}

\newcommand{\absval}[1]{\left| {#1} \right|}
\newcommand{\sset}[1]{\left\lbrace {#1} \right\rbrace }

\newcommand{\ssum}[2]{\displaystyle\sum\limits_{#1}^{#2}}

\def\d{\textrm{d}}

\newcommand{\nnorm}[2]{\absval{\absval{{#1}}}_{#2}}

\newcommand{\expp}[1]{\mathrm{E} \parent{{#1}}}

\newcommand{\kaustaddress}[0]{King Abdullah University of Science and Technology, Thuwal, Kingdom of Saudi Arabia}

\newcommand{\indicator}{\mathbf{1}}

\newcommand{\levy}{L\'evy }

\newcommand{\rz}{\R \setminus \sset{0}}

\newcommand{\linf}[1]{L^\infty_{#1}}
\newcommand{\norminf}[2]{\left\|#1\right\|_{\linf{#2}}}
\newcommand{\cl}{C \parent{\lambda}}

\pgfmathdeclarefunction{lg10}{1}{%
    \pgfmathparse{ln(#1)/ln(10)}%
}

\title[FFT and PIDEs]{Error analysis in Fourier methods for option pricing}
\author{Fabi\'an Crocce}
\address{\kaustaddress}
\email{fcrocce@gmail.com}

\author{Juho H\"app\"ol\"a}
\address{\kaustaddress}
\email{juho.happola@iki.fi}

\author{Jonas Kiessling}
\address{Jonas' address}
\email{jonkie@kth.se}

\author{Ra\'ul Tempone}
\address{\kaustaddress}
\email{raul.tempone@kaust.edu.sa}

\date{\today}    
\begin{document}

\begin{abstract}
We provide a bound for the error committed when using 
a Fourier method to price European options when the 
underlying follows an exponential \levy dynamic.
The price of the option is described by a
partial integro-differential
equation (PIDE). Applying a Fourier transformation to the PIDE 
yields an ordinary differential equation (ODE) that can be solved
analytically in terms of the characteristic exponent of the
\levy process. 
Then, a numerical inverse Fourier transform allows us 
to obtain the option price. 
We present a bound for the error and use this bound
to set the parameters for the numerical method. We analyse the
properties of the bound and demonstrate the minimisation of the bound
to select parameters for a numerical Fourier transformation method
to solve the option price efficiently.
\end{abstract}

\maketitle

\section{Introduction}
\label{section:Introduction}

\levy processes form a rich field within mathematical
finance. They allow modelling of asset prices with possibly discontinuous dynamics. An early and probably the best known model involving
a \levy process is
the \cite{Merton1976125} model, which generalises the \cite{black1973pricing} model. More recently, we have seen more complex models allowing
for more general dynamics of the asset price. Examples of such models include
the \cite{kou2002jump} model (see also \cite{Dotsis20073584}), the Normal Inverse Gaussian
model (\cite{barndorff1997normal, rydberg1997normal}),
the Variance Gamma model (\cite{madan1990variance, madan1998variance}),
and the Carr-Geman-Madan-Yor (CGMY) model
(\cite{cgmy_fine,carr2003stochastic}). For a good exposition on jump processes in finance we refer to \cite{conttankov}
(also see \cite{raible2000levy} and \cite{eberlein2001application}).

Prices of European options whose underlying asset is driven
by the \levy process are solutions to partial integro-differential
Equations (PIDEs) (\cite{nualart2001backward,briani_chioma_natalini,almendral2005numerical,kiessling2011diffusion})
that generalise
the Black-Scholes equation by incorporating a
non-local integral term to account for the discontinuities in
the asset price. This approach has also been extended to
cases where the option price features path dependence,
for instance in \cite{boyarchenko2002barrier} \cite{d2004penalty} and \cite{lord2008fast}.

The \levy-Khintchine formula provides an explicit 
representation of the characteristic function of a 
\levy process (cf, \cite{tankov2004financial}).
As a consequence, one can derive an exact expression for the 
Fourier transform of the solution of the relevant PIDE.
Using the inverse fast Fourier transform (iFFT) method, one
may efficiently compute the option price for a range of asset prices
simultaneously.
Furthermore, in the case of European call options, one may
use the duality property presented by \cite{dupire1997pricing}
and iFFT to efficiently compute option prices for a wide range of strike prices.

Despite the popularity of Fourier methods for option pricing, few works
can be found on the error analysis and related parameter selection for these methods. 
A bound for the error not only provides an interval for the precise
value of the option, but also suggests a method to select the parameters of the numerical method.
An important
work in this direction is the one by \cite{lee2004option} in which several
payoff functions are considered for a rather general set of models, whose
characteristic function is assumed to be known.
\cite{fenglinetsky} presents the framework and theoretical approach
for the error analysis, and establishes polynomial convergence rates
for approximations of the option prices.
For a more
contemporary review on the error committed in various
FT-related methods we refer the reader to \cite{boyarchenko2011new},
that extends the classical \emph{flat} Fourier methods by deforming the integration countours
on the complex plane, studying discretely monitored barrier options
studied in \cite{de2014pricing}.

In this work, we present a methodology for studying
and bounding  the error committed
when using FT methods to compute option prices.
We also provide a systematic way of choosing 
the parameters of the numerical method, in a way that minimises
the strict error bound, thus guaranteeing adherence to a pre-described
error tolerance. 
We focus on exponential \levy processes that may be of either
diffusive or pure jump. Our contribution is to derive a strict error bound for a Fourier
transform method when pricing options under risk-neutral 
\levy dynamics. We derive a simplified bound that separates the contributions of
the payoff and of the process in an easily processed and extensible
product form that is independent of the asymptotic behaviour of the
option price at extreme prices and at strike parameters.
We also provide a proof for the existence of optimal
parameters of the numerical computation that minimise the presented
error bound.
When comparing our work with Lee's work we find that Lee's work is more general than ours in that he studies a wider range
of processes, on the other hand, our results apply to a larger class of payoffs.
On test examples of practical relevance, we also find that the bound presented produces comparable or better results than the ones previously presented in the literature, with acceptable computational cost.

The paper is organised in the following sections:
In Section \ref{section:Method}
we introduce the PIDE setting in the context of
risk-neutral asset pricing; we show the Fourier representation of the relevant
PIDE for asset pricing with \levy processes and use that
representation for derivative pricing.
In Section \ref{section:Bound}
we derive a representation
for the numerical error and divide it into
quadrature and cutoff contributions. We also describe the
methodology for choosing numerical
parameters to obtain minimal error bounds
for the FT method.
The derivation is supported by numerical examples
using relevant test cases with both diffusive and
pure-jump \levy processes in Section \ref{section:Numerics}.
Numerics are followed by conclusions in Section \ref{section:Conclusion}.

\section{Fourier method for option pricing}

\label{section:Method}
Consider an asset whose price at time $t$ is modelled 
by the stochastic process $S=(S_t)$ defined by $S_t=S_0 \expfs{X_t}$,
where $X=(X_t) \in \R$ is assumed to be
a \levy process whose jump measure $\nu$ satisfies
\begin{align} \label{eq:nuhyp}
\int_{\rz} \min\{y^2,1\} \nu \parent{dy} < \infty
\end{align}

Assuming the risk-neutral dynamic for $S_t$, the price at time $t=T-\tau$ of a European option
with payoff $G$ and maturity time $T$ is given by
\begin{align*}
\Pi(\tau,s)=\expfs{-r\tau} \expp{G\parent{S_T}|S_{T-\tau}=s}
\end{align*}
where $r$ is the short rate that we assume to be constant and $\tau\colon 0\leq\tau\leq T$ is the time to maturity. 
Extensions to non-constant deterministic short rates are straightforward.

The infinitesimal generator of a \levy
process $X$ is given by \citep[see][]{applebaum2004levy}
\begin{align}
\lx f(x)&\equiv 
\lim_{h\to 0} \frac{\expp{f(X_{t+h})|X_t=x}-f(x)}h \notag \\
&= \gamma f'(x)+\frac12\sigma^2f''(x)+\int_{\rz} \left(f(x+y)-f(x)-y1_{|y|\leq 1} f'(x)\right)\nu(\d y)
\end{align}
where $(\gamma,\sigma^2,\nu)$ is the characteristic triple of the \levy process.
The risk-neutral assumption on $(S_t)$ implies
\begin{equation} \label{eq:riskneutralexpint}
\int_{\absval{y}>1}\expfs{y} \nu(\d y)<\infty
\end{equation}
and fixes the drift term (see \cite{kiessling2011diffusion}) $\gamma$ of the \levy process to
\begin{align} 
\label{drift_equation}
\gamma = r-\frac12 \sigma^2 - \int_{\rz} \parent{\expfs{y}-1-y1_{|y|\leq 1} }\nu(\d y)
\end{align}
Thus, the infinitesimal generator of $X$ may be written under the risk-neutral assumption as
\begin{align}
\lx f(x)&= \left(r-\frac{\sigma^2}2\right) f'(x)+\frac{\sigma^2}2 f''(x)+\int_{\rz} \left(f(x+y)-f(x)-(\expfs{y}-1)f'(x)\right)\nu(\d y)
\end{align}
Consider $g$ as the reward function in log prices (ie, defined by 
$g(x)=G(S_0 \expfs{x})$). Now, take $f$ to be defined as
\begin{align*}
f \parent{\tau,x} \equiv \expp{g \parent{X_{T}}|X_{T-\tau}=x}
\end{align*}
Then $f$ solves the following PIDE:
\begin{equation*}
\begin{cases}
\partial_{\tau}f(\tau,x)&=\lx f(\tau,x) \\
f(0,x)&=g(x), \qquad \qquad \parent{\tau,x}  \in [0,T] \times \R
\end{cases}
\end{equation*}
Observe that $f$ and $\Pi$ are related by
\begin{equation}
\Pi(\tau,S_0 \expfs{x})=\expfs{-r \tau}f(\tau,x)
\end{equation}
Consider a damped version of $f$ defined by $\fa(\tau,x)=\expfs{-\damp x}f(\tau,x)$; we see that $\partial_{\tau} \fa = \expfs{-\damp x} \lx f\parent{\tau,x} $.

There are different conventions for the Fourier transform. Here we consider the operator $\F$ such that 
\begin{align}
\F[f] \parent{\omega} \equiv \int_{\R}\expfs{i\omega x} f \parent{x} \d x
\label{eq:fourierDefinition}
\end{align}
defined for functions $f$ for which the previous integral is convergent. We also use $\ft{f} \parent{\omega}$ as a shorthand notation of $\F[f]\parent{\omega}$.
To recover the original function $f$, we define the inverse Fourier transform as 
$$\ift{f}(x)=\frac{1}{2\pi}\int_{\R} \expfs{-i\omega x}f(\omega) \d \omega$$
We have that $\ift{\ft{f}}(x)=f(x)$.

Applying $\F$ to $\fa$ we get $\ft{\fa}(\omega)=\ft{f} \parent{\omega + i \damp}$.
Observe also that the Fourier transform applied to $\lx f(\tau,x)$ gives $\charExp{-i\omega}\ft{f} \parent{\tau,\omega}$,
where $\charExp{\cdot}$ is the characteristic exponent of the process $X$, which satisfies $\expp{\expfs{zX_t}}=\expfs{t\charExp{z}}$.
The explicit expression for $\charExp{\cdot}$ is 
\begin{equation}
\label{characteristicDefinition}
\charExp{z}
=\parent{r-\frac{\sigma^2}2} z +\frac{\sigma^2}2 z^2 +\int_{\R} \left(\expfs{z y}-1-(\expfs{y}-1)z\right)\nu(\d y) \\[.5em]
\end{equation}
From the previous considerations it can be concluded that 
\begin{equation}
\partial_{\tau} \ft{\fa} = \charExp{\damp -i \omega} \ft{f} \parent{\omega - i \damp}
\end{equation}
Now $\ft{f}\parent{\omega - i \damp} = \ft{\fa} \parent{\omega}$ so $\ft{\fa}$ satisfies the following ODE
\begin{equation}
\begin{cases}
\frac{\partial_{\tau}\ft{\fa}\parent{\tau,\omega}}{\ft{\fa}(\tau,\omega)}&=\charExp{\damp -i \omega}\\[.5em]
\ft{\fa} \parent{0,\omega}&=\ft{\ga}\parent{\omega}
\end{cases}
\end{equation}

Solving the previous ODE explicitly, we obtain
\begin{equation}
\ft \fa \parent{\tau,\omega}
= \expfs{\tau \charExp{\damp-i \omega} } \ft{\ga} \parent{\omega}
\label{eq:theODE}
\end{equation}
Observe that the first factor in the right-hand side in the above 
equation is $\expp{\expfs{\parent{\damp-i \omega}X_{\tau}}}$, (ie, $\charFun{1}{-i\damp-\omega}$), where $\charFun{\tau}{\cdot}$ denotes the characteristic function of the
random variable $X_{\tau}$
\begin{align}
\charFun{\tau}{\omega} \equiv \expp{\tau \charExp{i \omega}} 
\end{align}

Now, to obtain the value function we employ
the inverse Fourier transformation, to obtain
\begin{equation}
\fa(\tau,x)=\ift{\ft{\fa}}(\tau,x)=\frac{1}{2\pi}\int_{\R} \expfs{-i\omega x} \ft{\fa}(\tau,\omega)d\omega
\label{eq:fadefinition}
\end{equation}
or
\begin{equation}
\fa(\tau,x)=\frac{1}{\pi}\int_{0}^{+\infty} \real{\expfs{-i\omega x} \ft{\fa}(\tau,\omega)}d\omega
\label{eq:fadefinition2}
\end{equation}

As it is typically not possible to compute the inverse Fourier transform analytically,
we approximate it by
discretising and truncating the integration domain
using trapezoidal quadrature \eqref{eq:fadefinition}.
Consider the following approximation:
\begin{align}
\label{eq:fdiscrete}
 \fadisc{\step,n} (\tau,x)&= \frac{\step}{2\pi} \sum_{k=-n}^{n-1} \expfs{-i{\parent{k+\frac12}\step} x} \ft{\fa}\parent{\tau,{\parent{k+\frac12}\step}}\\
 &=\frac{\step}{\pi} \sum_{k=0}^{n-1} \real{\expfs{-i{\parent{k+\frac12}\step} x} \ft{\fa}\parent{\tau,{\parent{k+\frac12}\step}}}
\end{align} 
Bounding and consequently minimising the error
in the approximation of $f(\tau,x)$ by $$\fdisc{\step,n}(\tau,x) \equiv \expfs{\damp x}\fadisc{\step,n}(\tau,x)$$
is the main focus of this paper and will be addressed in the following section.

\begin{remark}
Although we are mainly concerned with option pricing when 
the payoff function can be damped in order to guarantee
regularity in the $L^1$ sense, we note here that our main 
results are naturally extendable to include the Greeks of the option. 
Indeed, we have by \eqref{eq:theODE} that 
\begin{align}
 f \parent{t,x} = \frac{1}{2 \pi} \int_{\mathbb R} \expfs{\parent{\damp - i \omega}x} \ft \fa \parent{\tau, \omega} \d \omega
\end{align}
so the Delta and Gamma of the option equal
\begin{align}
\Delta \parent{t,x} &\equiv \frac{\partial f \parent{t,x}}{\partial x} = \frac{1}{2 \pi} \int_{\mathbb R} \parent{\damp-i \omega} \expfs{\parent{\damp - i \omega}x} \ft \fa \parent{\tau, \omega} \d \omega
\\
\Gamma \parent{t,x} &\equiv \frac{\partial^2 f \parent{t,x}}{\partial x^2} = \frac{1}{2 \pi} \int_{\mathbb R} \parent{\damp-i \omega}^2 \expfs{\parent{\damp - i \omega}x} \ft \fa \parent{\tau, \omega} \d \omega
\end{align}
Because the expressions involve partial derivatives with respect to only $x$,
the results in this work are applicable for the computation of $\Delta$ and $\Gamma$
through a modification of the payoff function:
\begin{align}
\ft g_{\damp, \Delta} \parent{\omega} =& \ft g_{\damp} \parent{\omega} \parent{\damp - i \omega}
\\
\ft g_{\damp, \Gamma} \parent{\omega} =& \ft g_{\damp} \parent{\omega} \parent{\damp-i \omega}^2
\end{align}
When the Fourier space payoff function manifests
exponential decay, the introduction of a coefficient that is polynomial in $\omega$
does not change the regularity of $\ft g$ in a way that would significantly change the following
analysis. Last, we note that since we do our analysis for PIDEs on a mesh of
$x$'s, one may also compute the option values in one go and obtain the Greeks with
little additional effort using a finite difference approach for the derivatives.
\end{remark}

\subsection{Evaluation of the method for multiple values of $x$ simultaneously}

The Fast Fourier Transform (FFT) algorithm provides an efficient 
way of computing \eqref{eq:fdiscrete} for an equidistantly
spaced mesh of values for $x$ simultaneously.
Examples of works that consider this widely extended tool are \cite{lord2008fast,jackson2008fourier,hurd2010fourier} and \cite{schmelzle2010option}.

Similarly, one may define the Fourier frequency $\omega$
as the conjugate variable of some external parameter on which
the payoff depends. Especially, for the practically relevant case
of call options, we can denote the log-strike as $k$ and treat
$x$ as a constant and write:
\begin{align}
\kft f_{k,\damp} \parent{\omega} \equiv
\int_\R \expfs{\parent{\damp +i \omega }k } f_k \parent{x} \d k
\end{align}

Using this convention, the time dependence is given by
\begin{align}
\label{eq:kftDef}
\kft f_{k,\damp} \parent{\tau,x} = 
\frac{\expfs{\parent{i \omega + \damp + 1} x} \charFun{\tau}{\omega - i \parent{\damp +1} }}{\parent{i \omega + \damp} \parent{i \omega + \damp + 1}}
\end{align}
contrasted with the $x$-space solution
\begin{align}
\label{eq:xftDef}
\ft f_{k,\damp} \parent{\tau,x} = 
\frac{\expfs{\parent{i \omega - \damp + 1} k} \charFun{\tau}{\omega + i \damp }}{\parent{i \omega + \damp} \parent{i \omega + \damp + 1}}
\end{align}
We note that for call option payoff to be in $L^1$,
we demand that $\damp$ in \eqref{eq:kftDef} is positive.
Omitting the exponential factors that contain the $x$ and $k$
dependence in \eqref{eq:kftDef} and \eqref{eq:xftDef} respectively,
we have that one can arrive from \eqref{eq:kftDef} to \eqref{eq:xftDef}
using the mapping $\damp \mapsto -\damp-1 $. Thanks to this,
much of the analysis regarding the $x$-space transformation
generalises in a straightforward manner to the $k$-space transform.

\section{Error bound}
\label{section:Bound}

The aim of this section is to compute a bound of the error when approximating 
the option price $f(\tau,x)$ by $\fadisc{\step,n} (\tau,x)$,
defined in \eqref{eq:fdiscrete}.
Considering
\begin{equation} \label{eq:serie}
\fadisc{\step}(\tau,x)= \frac{\step}{2\pi} \sum_{k\in \Z} \expfs{-i{\parent{k+\frac12}\step} x} \ft{\fa}\parent{\tau,{\parent{k+\frac12}\step}}
\end{equation}
the total error $\Error$ can be split into a sum of two terms: the quadrature and truncation errors.
The former is the error from the approximation of the integral in \eqref{eq:fadefinition}
by the infinite sum in \eqref{eq:serie}, while the latter is due to the truncation of the infinite sum. 
Using triangle inequality, we have
\begin{align} \label{eq:errordef}
\Error &:=\absval{f(\tau,x)-\fdisc{\step,n}(\tau,x)}\leq \EQ + \EF
\end{align}
with
\begin{align*}
\EQ &= \expfs{\damp x}\absval{\fa(\tau,x)-\fadisc{\step}(\tau,x)}\\
\EF &= \expfs{\damp x} \absval{\fadisc{\step}(\tau,x)-\fadisc{\step,n}(\tau,x)}
\end{align*}
Observe that each $\Error$, $\EQ$ and $\EF$ depend on
three kinds of parameters:
\begin{itemize}
 \item Parameters underlying the model and payoff such as volatility and strike price. We call these \emph{physical parameters.}
 \item Parameters relating to the numerical scheme such as $\damp$ and $n$.
 \item \emph{Auxiliary parameters} that will be introduced in the process of deriving the error bound. These parameters
 do not enter the computation of the option price, 
 but they need to be chosen appropriately to have as tight a bound as possible. 
\end{itemize}

We start by analysing the quadrature error.
\subsection{Quadrature error}

Denote by $\strip{a}$, with $a>0$, the strip of width $2a$ around the real line:
$$\strip{a} \equiv \left\{z\in \C \colon \absval{\im{z}}<a \right\}$$
The following theorem presents conditions under which the quadrature error goes to zero at a spectral rate as $\step$ goes to zero.
Later in this section, we discuss simpler conditions to verify the hypotheses and analyse in more detail the case 
when the process $X$ is a diffusive process or there are ``enough small jumps."

\begin{theorem}\label{theo:quadrature} 
Assume that for $a > 0$:
\begin{itemize}
\item[H1.] the characteristic function of the random variable $X_1$ has an analytic extension to the set $$\strip{a}-\damp i \equiv \left\{z\in \C \colon \absval{\im{z}+\damp}<a\right\}$$
\item[H2.] the Fourier transform of $\ga(x)$ is analytic in the strip $\strip{a}$ and
\item[H3.] there exists a continuous function $\gamma \in L^1(\R)$ such that  $\absval{\ft{\fa}(\tau,\omega+i \beta)}<\gamma(\omega)$ for all $\omega\in \R$ and for all $\beta \in [-a,a]$
\end{itemize}
Then the quadrature error is bounded by
\begin{equation*}
\ErrorQuadrature \leq \expfs{\damp x} \frac{M_{\damp,a}(\tau,x) }{2\pi\parent{\expfs{2\pi a/\step}-1}}
\end{equation*}
where $M_{\damp,a}(\tau,x)$ is given by
\begin{align}\label{eq:defM}
M_{\damp,a}(\tau,x)&:=\sum_{\beta \in \{-a,a\}}\int_{\R}\absval{\expfs{-i\parent{\omega + i\beta} x}\ft{\fa}(\tau,\omega+i \beta)} \d \omega
\end{align}
$M_{\damp,a}(\tau,x)$ equals the Hardy norm (defined in \eqref{eq:hardynorm}) of the function $\omega \mapsto \expfs{-i\parent{\omega + i\beta} x}\ft{\fa}(\tau,\omega+i \beta)$, which
is finite.
\end{theorem}

The proof of Theorem \ref{theo:quadrature} is an application of Theorem 3.2.1 in \cite{stenger2012numerical},
whose relevant parts we include for ease of reading. Using the notation in \cite{stenger2012numerical}, $H^1_{\strip{a}}$ is the family of functions $w$ that are
 analytic in $\strip{a}$ and such that 
\begin{equation}\label{eq:hardynorm}
\nnorm{w}{H^1_{A_a}}:= \lim_{\varepsilon\to 0} \int_{\partial \strip{a}(\varepsilon)} \absval{w(z)}d\absval{z} < \infty
\end{equation}
 where 
 $$\strip{a}(\varepsilon) = \left\{z\in \C\colon \absval{\real{z}}<\frac1{\varepsilon},\ \absval{\im{z}} < a \parent{1-\varepsilon}\right\}$$
\begin{lemma}[Theorem 3.2.1 in \cite{stenger2012numerical}]\label{lem:stenger}
Let $w \in H^1_{A_a}$,
then define:
\begin{align}
I \parent{w} &= \int_\R w \parent{x} \d x \\
J \parent{w,h} &= h \ssum{j=-N}{N} w \parent{jh} \\
\zeta \parent{w,h} &= I \parent{w}- J \parent{w,h}
\end{align}
then
\begin{align}
\absval{\zeta \parent{w,h}} \leq \frac{e^{-\frac{\pi a}{h}} \nnorm{w}{H^1_{A_a}}}{2 \mathrm{sinh} \parent{\frac{\pi a}{h}}}
\end{align}
\end{lemma}

\begin{proof}[Proof of Theorem \ref{theo:quadrature}]
First observe that H1 and H2 imply that the function $w(z)=\expfs{-i x z + \step/2}\ft{\fa}(\tau,z+ \step/2)$ is analytic in $\strip{a}$. 
H3 allows us to use dominated convergence theorem to prove that $\nnorm{w}{H^1_{A_a}}$ is finite and coincides with $M_{\damp,a}(\tau,x)$. Applying Lemma \ref{lem:stenger} 
the proof is completed.
\end{proof}

Regarding the hypotheses of Theorem \ref{theo:quadrature}, the next propositions provide simpler conditions that imply H1 and H2 respectively.

\begin{proposition} \label{prop:analy1}
If $\alpha$, $a$ and $\nu$ are such that
\begin{equation}
\label{eq:alterH1}
\int_{y>1} \expfs{\parent{\damp+a}y}\nu(\d y) < \infty \quad \text{and}\quad \int_{y<-1} \expfs{\parent{\damp-a}y}\nu(\d y) < \infty
\end{equation}
then H1 in Theorem \ref{theo:quadrature} is fulfilled.

\begin{proof}
Denoting by $\charFun{1}{\cdot}$ the characteristic function of $X_1$, we want to prove that $z\mapsto \charFun{1}{z+\damp i}$ is analytic in $\strip{a}$. Considering that $\charFun{1}{z+\damp i}=\expfs{\charExp{iz-\damp}}$, the only non-trivial part of the proof is to verify that 
\begin{equation}
z\mapsto \int_{} p(z,y)\nu(\d y)
\end{equation}
is analytic in $\strip{a}$,
where $p\colon \strip{a} \times \R \to \C$ is given by
$$p(z,y)=\expfs{y \parent{i z-\damp}}-1-(\expfs{y}-1)\parent{i z-\damp}$$
To prove this fact, we show that we can apply the main result and the only theorem in \cite{mattner2001complex}, 
which, given a measure space $(\Omega,\mathcal{A},\mu)$ and an open subset $G\subseteq \C$, 
ensures the analyticity of $\int f(\cdot,\omega) d\mu(\omega)$, provided that $f\colon G\times \Omega \to \C$ satisfies: $f(z,\cdot)$ is $\mathcal{A}$-measurable for all $z\in G$; $f(\cdot,\omega)$ is holomorfic for all $\omega \in \Omega$; and $\int \absval{f(\cdot,\omega)} d\mu(\omega)$ is locally bounded.
In our case we consider the measure space to be $\R$ with the Borel $\sigma$-algebra and the Lebesgue measure, $G=\strip{a}$ and $f=p$. It is clear that 
$p(x,\cdot)$ is Borel measurable and $p(\cdot,y)$ is holomorphic. It remains to verify that  
\begin{equation*}
z \mapsto \int_{\R^*} \absval{p(z,y)}\nu(\d y)
\end{equation*} 
is locally bounded. 
To this end, we assume that $\real{z}<b$ (and, since $z\in \strip{a}$, $\im{z}<a$) and split the integration domain in $\absval{y}>1$ and ${0<\absval{y}\leq 1}$ to prove that both integrals are uniformly bounded.

Regarding the integral in ${\absval{y}> 1}$, we observe that
\begin{equation}
\absval{p(z,y)}\leq \expfs{y(\damp+\im{z})}+1+(\expfs{y}+1)(\damp+a+b)
\end{equation}
for $y<-1$ we have $\expfs{y(\damp+\im{z})}<\expfs{y(\damp-a)}$ while for $y>1$ we have $\expfs{y(\damp+\im{z})}<\expfs{y(\damp+a)}$. 
Using the previous bounds and the hypotheses together with \eqref{eq:nuhyp} and \eqref{eq:riskneutralexpint}, we obtain the needed bound.

For the integral in ${0<\absval{y}\leq 1}$, observe that, denoting $f(z,y)=\absval{p(z,y)}$, we have, $f(z,0)=0$ for every $z$, $\partial_y f(z,0)=0$ for every $z$, and $\absval{\partial_{yy}f(z,y)}< c$ for $z\in \strip{a}, \real{z}<b, \absval{y}<1$. From these observations we get that the McLaurin polynomial of degree one of $y\mapsto f(z,y)$ is null for every $z$, and we can bound $f(z,y)$ by the remainder term, which, in our region of interest, is bounded by $\frac{c}{2} y^2$, obtaining
\begin{equation}
\int_{0<\absval{y}\leq 1} \absval{p(z,y)}\nu(\d y)\leq \frac{c}2 \int_{0<\absval{y}\leq 1}y^2\nu(\d y)
\end{equation}
which is finite by hypothesis on $\nu$, which finishes the proof.
\end{proof}
\end{proposition}

\begin{proposition}
If for all $b<a$, the function $x\mapsto \expfs{b\absval{x}}{\ga(x)}$ is in $L^2(\R)$ then H2 in Theorem \ref{theo:quadrature} is fulfilled.
\end{proposition}
\begin{proof}
The proof is a direct application of Theorem IX.13 in \cite{reedsimon}
\end{proof}
We now turn our attention to a more restricted class of Levy processes. 
Namely, processes such that either $\sigma^2>0$ or there exists $\lambda \in (0,2)$ such that $\cl$ defined in \eqref{eq:PureJumpC} is strictly positive. 
For this class of processes, we can state our main result explicitly in terms of the characteristic triplet. 

Given $\lambda \in (0,2)$, define $\cl$ as
\begin{equation}
\label{eq:PureJumpC}
\cl  = 
\textrm{inf}_{\kappa > 1} \left\{\kappa^{\lambda} 
 \int_{0<|y| < \frac{1}{\kappa}} y^2 \nu(\d y)\right\}
\end{equation}
Observe that $\cl\geq 0$ and, by our assumptions on the jump measure $\nu$, $\cl$ is finite. Furthermore, if $\lambda \in (0,2)$ is such
that
\begin{equation}
\label{eq:pureJumpAssumption}
\liminf_{\epsilon \downarrow 0} \frac{1}{\epsilon^{\lambda}}
 \int_{0<|y| < \epsilon} y^2 \nu(\d y) > 0
\end{equation}
then $\cl>0.$ To see this, observe that \eqref{eq:pureJumpAssumption} implies the existence of $\epsilon_0$ such that 
\begin{equation*}
\inf_{\epsilon \leq \epsilon_0}\left\{ \frac{1}{\epsilon^{\lambda}}
 \int_{0<|y| < \epsilon} y^2 \nu(\d y)\right\} > 0
\end{equation*}
If $\epsilon_0<1$ observe that
\begin{equation*}
\inf_{\epsilon_0 \leq \epsilon \leq 1}\left\{ \frac{1}{\epsilon^{\lambda}}
 \int_{0<|y| < \epsilon} y^2 \nu(\d y)\right\} \geq  \int_{0<|y| < \epsilon_0} y^2 \nu(\d y) > 0
\end{equation*}
where for the first inequality it was taken into account that $\frac{1}{\epsilon^{\lambda}}\geq 1$ and that the integral is increasing with $\epsilon$.
Combining the two previous infima and considering $\absval{\kappa}=\frac{1}{\epsilon}$ we get that $\cl>0.$

Furthermore, we note that for a \levy model with finite jump
intensity, such as the Black-Scholes and Merton models that satisfy the first of our
assumption, $\cl=0$ for all $\lambda \in (0,2)$.

\begin{theorem} \label{teo:boundM}
Assume that: $\damp$ and $a$ are such that \eqref{eq:alterH1} holds; $\ft{\ga} \in \linf{\strip{a}}$; and either $\sigma^2>0$ or $\cl>0$ for some $\lambda\in (0,2)$.
Then the quadrature error is bounded by
\begin{equation*}
\ErrorQuadrature \leq \expfs{\damp x} \frac{\tilde{M}_{\damp,a}(\tau,x) }{2\pi\parent{\expfs{2\pi a/\step}-1}}
\end{equation*}
where 
\begin{equation}\label{eq:Mbound}
\tilde{M}_{\damp,a} \parent{\tau,x} 
=
\ssum{c\in \sset{-1,1}}{} 
\expfs{cax}
\expfs{\tau \Psi \parent{ca}} \absval{\ft{\ga} \parent{ca}}
\int_{\R}  \expfs{-\tau \left( \frac{\sigma^2}{2}\omega^2 + \frac{|\omega|^{2 - \lambda}}{4} \cl\indicator_{\absval{\omega}>1}\right) } \d \omega
\end{equation}
Furthermore, if $\sigma^2>0$ we have
\begin{equation}
\label{eq:Mboundsigpositive}
\tilde{M}_{\damp,a} \parent{\tau,x}  \leq \frac{\sqrt{2\pi}}{\sigma\sqrt{\tau}}\ssum{c\in \sset{-1,1}}{} 
\expfs{cax}
\expfs{\tau \Psi \parent{ca}} \absval{\ft{\ga} \parent{ca}}
\end{equation}

\begin{proof}
Considering $h_{\damp,a}(\tau,x,\omega)$ defined by
\begin{equation}
h_{\damp,a} \parent{\tau,x,\omega} = \ssum{c \in \sset{-1,1}}{} \absval{\expfs{-i\parent{\omega + ica} x}\ft{\fa}(\tau,\omega+i ca)} 
\end{equation}
we have that
$$M_{\damp,a}(\tau,x) = \int_{\R}h_{\damp,a} \parent{\tau,x,\omega} d\omega$$
On the other hand, for  $\beta\in (-a,a)$:
\begin{equation}
\absval{\expfs{-i\parent{\omega + i\beta} x}\ft{\fa}(\tau,\omega+i \beta)} = \expfs{\beta x} \absval{\ft{\fa}(\tau,\omega+i \beta)}=\expfs{\beta x}\absval{\expfs{\tau \charExp{\damp+\beta-i\omega}}}\absval{ \ft{\ga}\parent{\omega+i\beta}}
\end{equation}
For the factor involving the characteristic exponent we have
\begin{equation}
\absval{\expfs{\tau \charExp{\damp+\beta-i\omega}}}=\expfs{\tau \real{\charExp{\damp+\beta-i\omega}}}
\end{equation}

Now, observe that
\begin{align}
\notag \real{\charExp{ \damp +\beta - i \omega}}&= \parent{\damp + \beta} \parent{r-\frac{\sigma^2}{2}}+\frac{\sigma^2}{2}\parent{\parent{\damp +\beta}^2 - \omega^2} \\
\label{eq:boundingPsiPureJump} &\quad +\int_{\rz}\parent{ \expfs{\parent{\damp + \beta} y} \cos(-y\omega)-1-\parent{\damp + \beta} \parent{\expfs{y}-1}} \nu(\d y)
\end{align}

If $|\omega| \leq 1$ we bound $\cos(-y \omega)$ by 1, getting
\begin{align}
\notag \real{\charExp{ \damp +\beta - i \omega}} &\leq \parent{\damp + \beta} \parent{r-\frac{\sigma^2}{2}}+\frac{\sigma^2}{2}\parent{\parent{\damp +\beta}^2 - \omega^2} \\
\notag &\quad +\int_{\rz} \parent{\expfs{\parent{\damp + \beta} y} -1-\parent{\damp + \beta} \parent{\expfs{y}-1}} \nu(\d y)  \\
& = \charExp{\damp+\beta} -\frac{\sigma^2}{2}\omega^2 \label{eq:Psibound}
\end{align}

Assume $| \omega| > 1$.
Using that for $|x| < 1$ it holds that $\textrm{cos}(x) < 1 - x^2 / 4$, we can bound the first term of the integral in the 
following manner:
\begin{align}
\notag \int_{\rz} \expfs{\parent{\damp + \beta} y} \cos(y\omega) \nu(\d y) & \leq
\int_{0<|y| < 1/|\omega|} \expfs{\parent{\damp + \beta} y} \left( 1 - \omega^2 y^2 / 4 \right)  \nu(\d y) \\
\notag & \quad + 
\int_{|y| \geq  1/|\omega|} \expfs{\parent{\damp + \beta} y}  \nu(\d y) \\
\notag & 
\leq \int_{\rz} \expfs{\parent{\damp + \beta} y}  \nu(\d y) 
- \frac{|\omega|^{2 - \lambda}}{4} | \omega |^{\lambda} \int_{0<|y| < 1 / |\omega|} y^2 \nu(\d y)  \\
\label{eq:extraTermPureJump} & 
\quad 
\leq \int_{\rz} \expfs{\parent{\damp + \beta} y}  \nu(\d y) 
- \frac{|\omega|^{2 - \lambda}}{4} \cl
\end{align}
Inserting \eqref{eq:extraTermPureJump} back into \eqref{eq:boundingPsiPureJump} we get
\begin{align*}
\notag 
\real{\charExp{ \damp +\beta - i \omega}} & \leq \parent{\damp + \beta} \parent{r-\frac{\sigma^2}{2}}+\frac{\sigma^2}{2}\parent{\parent{\damp +\beta}^2 - \omega^2} \\
\notag &\quad + \int_{\rz} \parent{\expfs{\parent{\damp + \beta} y} -1-\parent{\damp + \beta} \parent{\expfs{y}-1}} \nu(\d y) \\
\notag & \quad \quad - \frac{|\omega|^{2 - \lambda}}{4} \cl    \\
& = \charExp{\damp+\beta} -\frac{\sigma^2}{2}\omega^2 - \frac{|\omega|^{2 - \lambda}}{4} \cl
\end{align*}

Taking the previous considerations and integrating in $\R$ with respect to $\omega$, we obtain \eqref{eq:Mbound}. 

Finally, observing that  $\cl\geq 0$ and bounding it by 0, the bound \eqref{eq:Mboundsigpositive} is obtained by evaluating the integral.
\end{proof}
\end{theorem}

\begin{remark} \label{remark:callPut}
In the case of call options, hypothesis H2
implies a dependence between the strip-width parameter
$a$ and damping parameter $\damp$. We have that the
damped payoff of the call option is in $L^1 \parent{\R}$
if and only if $\damp > 1$ and hence the appropriate choice
of strip-width parameter is given by $0<a< \damp-1$.
A similar argument holds for the case of put options, for
which the Fourier-transformed damped payoff is identical
to the calls with the distinction that $\damp <0$.
In such case, we require $a < -\damp$.

The case of binary options whose payoff has finite support
($G \parent{x} = \mathbf{1}_{[x_-,x_+]} \parent{x}$) 
we can set any $a \in \R$ (ie, no damping is needed at all and
even if such damping is chosen, it has no effect on the appropriate
choice of $a$). 
\end{remark}

\begin{remark}
The bound we provide for the quadrature error is naturally positive and increasing in $\step$.
It decays to zero at a spectral rate as $\step$ decreases to 0.
\end{remark}

\subsection{Frequency truncation error}
The frequency truncation error is given by
\begin{align*}
\ErrorFrequency &= \frac{\expfs{\damp x}\step}{\pi} \absval{\sum_{k=n}^{\infty} \real{\expfs{-i{\parent{k+\frac12}\step} x} \ft{\fa}\parent{\tau,{\parent{k+\frac12}\step}}}}
\end{align*}
If a function $c\colon (\omega_0,\infty) \to (0,\infty)$ satisfies
\begin{equation} \label{eq:condFunc}
\absval{\real{\expfs{-i{\parent{k+\frac12}\step} x} \ft{\fa}\parent{\tau,{\parent{k+\frac12}\step}}}} \leq c\parent{\parent{k+\frac12}\step}
\end{equation}
for every natural number $k$, then we have that 
\begin{align*}
\ErrorFrequency &\leq \frac{\expfs{\damp x}\step}{\pi} \sum_{k=n}^{\infty}\absval{\real{\expfs{-i{\parent{k+\frac12}\step} x} \ft{\fa}\parent{\tau,{\parent{k+\frac12}\step}}}} \\
& \leq \frac{\expfs{\damp x}\step}{\pi} \sum_{k=n}^{\infty} c\parent{\parent{k+\frac12}\step}
\end{align*}
Furthermore, if $c$ is a non-increasing concave integrable function, we get 
\begin{equation}
\ErrorFrequency \leq \frac{\expfs{\damp x}}{\pi} \int_{n\step}^{\infty} c(\omega) \d\omega
\label{eq:truncationErrorIntegral}
\end{equation} 
When $\ft{\ga}\in \linf{[\omega_0,\infty)}$ and either $\sigma^2>0$ or $\cl>0$, then the function $c$ in \eqref{eq:condFunc} can be chosen as
\begin{equation} \label{eq:cutoffparticularSimpler}
c(\omega)=\norminf{\ft{\ga}}{[\omega_0,\infty)}\expfs{\tau \charExp{\damp}}\expfs{-\tau \parent{\frac{\sigma^2}{2} \omega^2 +\frac{|\omega|^{2 - \lambda}}{4} \cl\indicator_{\absval{\omega}>1} }}
\end{equation}

To prove that this function satisfies \eqref{eq:condFunc} we can use the same bound we found in the proof of Theorem \ref{teo:boundM}, with $\beta=0$, to obtain
\begin{equation*}
\real{\charExp{ \damp - i \omega}}\leq  \charExp{\damp} - \frac{\sigma^2}{2}\omega^2  - \frac{|\omega|^{2 - \lambda}}{4} \cl \indicator_{\absval{\omega}>1}
\end{equation*}
from where the result is straightforward.

\subsection{Bound for the full error}

In this section we summarize the bounds obtained for the error under different 
assumptions and analyse their central properties. 

In general the bound provided in this paper are of the form
\begin{equation}\label{eq:fullErrorGeneral}
\est = \frac{\expfs{\damp x}}{\pi}\parent{ \frac{\bar{M}}{\expfs{2\pi a/\step}-1} +  \int_{n\step}^{\infty}  c\parent{\omega} \d \omega}
\end{equation}
where $\bar{M}$ is an upper bound of $M_{\damp,a}\parent{\tau,x}$ defined in \eqref{eq:defM} and $c$ is non-increasing, integrable and satisfies
\eqref{eq:condFunc}. Both $\bar{M}$ and $c$ may depend on the 
parameters of the model and the artificial parameters, but they are independent of $\step$ and $n$. Typically one can remove 
the dependence of some of the parameters, simplifying the expressions but obtaining
less tight bounds.

When analysing the behaviour of the bound one can observe that the term correspondent with the quadrature error decreases to zero
spectrally when $\step$ goes to 0. 
The second term goes to zero if $n\step$ diverges, but we are unable to determine
the rate of convergence without further assumptions.

Once an expression for the error bound is obtained, the problem of 
how to choose the parameters of the numerical method to minimise the bound arises, 
assuming a constraint on the computational effort one is willing to use. 
The computational effort of the numerical method depends
only on  $n$. For this reason we aim at finding the parameters that
minimise the bound for a fixed $n$.
The following result shows that the bound obtained, as a function of $\step$, 
has a unique local minimum, which is the global minimum.
\begin{proposition}
Fix $\alpha$, $a$, $n$, and $\lambda$ and consider the bound $\est$ as a
function of $\step$. There exists an optimal $\step^*\in [\frac{\omega_0}{n},\infty)$ such that $\est$ is
decreasing in $(\frac{\omega_0}{n},\step^*)$ and increasing in
$(\step^*,\infty)$; thus, a global minimum of $\est$ is attained at $\step^*$. 

Furthermore, the optimal $\step$ is either the only point in which 
$\step\mapsto p\parent{n\step,b}-c(n\step)$, with $p$ defined in \eqref{eq:defp}, changes sign, or $\step=\frac{\omega_0}{n}$ if $p\parent{\omega_0,b}-c(\omega_0)>0$.

\begin{proof}
Let us simplify the notation by calling $y=n\step$, $b=2 \pi a n$ and $\tilde{\E}=\pi \expfs{-\damp x}\est$. We want to prove the existence of $y^*\colon y^*\geq \omega_0$ such that $\tilde{\E}(y)$ is decreasing for $\omega_0<y< y^*$ and increasing for $y> y^*$. We have 
$$\tilde{\E}(y)=\frac{\bar{M}}{\expfs{b/y}-1} +  \int_{y}^{\infty} c(\omega) \d\omega.$$
The first term is differentiable with respect to $y$ and goes to 0 if $y\to 0^+$. This allows us to express it as an integral of its derivative. We can then express $\tilde{\E}(y)$ as 
$$\tilde{\E}(y)=\tilde{\E}(\omega_0) + \int_{\omega_0}^y \parent{\frac{b \bar{M} \expfs{b/\omega}}{\parent{\expfs{b/\omega}-1}^2 \omega^2}-c(\omega)}d(\omega)$$
The first term on the right-hand side of the previous equation is constant. Now we move on to proving that the integrand is increasing with $y$ and it is positive if $y$ is large enough. Denote by
\begin{equation}\label{eq:defp}
p(y,b)=\frac{b \bar{M} \expfs{b/y}}{\parent{\expfs{b/y}-1}^2 y^2}
\end{equation}
Taking into account that $c$ is integrable, we can compute the limit of the integrand in $\infty$, obtaining
\begin{equation*}
\lim_{y \to +\infty}p(y,b)-c(y)=\frac{\bar{M}}{b}>0
\end{equation*} 
Let us prove that $p(y,b)$ is increasing with $y$ for all $b>0$, which renders $p(y,b)-c(y)$ also increasing with $y$. The derivative of $p$ with respect to $y$ is given by
$$\partial_{y} p(y,b)=\frac{b\bar{M}\expfs{b/y}\parent{(b/y)\expfs{b/y}-2\expfs{b/y}+b/y+2}}{y^3\parent{\expfs{b/y}-1}^3}$$
in which the denominator and the first factor in the numerator are clearly positive. To prove that the remainder factor is also positive, observe that 
$x\expfs{x}-2\expfs{x}+x+2>0$ if $x>0$.

\end{proof}
\end{proposition}

\subsection{Explicit error bounds}
\label{ss:bounds}
In the particular case when either $\sigma^2>0$ or 
$\cl>0$ for some $\lambda \in (0,2)$ we can give an explicit version 
of \eqref{eq:fullErrorGeneral}. Substituting $M$ by $\tilde{M}$ defined in Theorem \ref{teo:boundM} and $c$ by the function given in \eqref{eq:cutoffparticularSimpler} we obtain
\begin{align}
\label{eq:fullErrorExpression}
\est &= \estQuadrature + \estFrequency
\end{align}
where
\begin{align}
\label{eq:boundQ} \estQuadrature &= 
\ssum{c\in \sset{-1,1}}{} 
\frac{
\expfs{\damp x}
\expfs{cax}
\expfs{\tau \Psi \parent{ca}} \absval{\ft{\ga} \parent{ca}}}{\pi \parent{ \expfs{\frac{2 \pi a}{\step}} -1}}
\int_{\R}  \expfs{-\tau \left( \frac{\sigma^2}{2}\omega^2 + \frac{|\omega|^{2 - \lambda}}{4} \cl\indicator_{\absval{\omega}>1}\right) } \d \omega
\\[.5em]
\label{eq:boundF}
\estFrequency &= \frac{\expfs{\damp x}}{\pi}
\norminf{\ft{\ga}}{\R}
\expfs{\tau \charExp{\damp}} \int_{n\step}^{\infty} 
\expfs{-\tau \parent{\frac{\sigma^2}{2} \omega^2+\frac{|\omega|^{2 - \lambda}}{4} \cl \indicator_{\absval{\omega}>1}} } d\omega
\end{align}
This reproduces the essential features of Theorem 6.6
in \cite{fenglinetsky},
the bound \eqref{eq:boundF} can be further improved by substituting $\norminf{\ft{\ga}}{\R}$ by $\norminf{\ft{\ga}}{[n\step,\infty)}$.

\begin{remark}\label{rem:decuople}
Observe that the bound of both the quadrature and the cutoff error is given
by a product of one factor that depends exclusively on the payoff and 
another factor that depends on the asset dynamic. This property makes it easy to evaluate
the bound for a specific option under different dynamics of the asset price.
In Subsection \ref{subsec:call} we analyse the terms
that depend on the payoff function for the particular case of call options.
\end{remark}

\begin{remark}\label{rem:asynonasy}
From \eqref{eq:boundQ} it is evident that the speed of the exponential convergence
of the trapezoidal rule for analytic functions is dictated by the width of the strip in which
the function being transformed is analytic. Thus, in the limit of small error tolerances,
it is desirable to set $a$ as large as possible to obtain optimal rates.
However, non-asymptotic error tolerances are often practically relevant and in these
cases the tradeoff between optimal rates and the constant term $\absval{\ft{\ga}}$
becomes non-trivial. As an example, for the particular case of the Merton model,
we have that any finite value of $a$ will do. However, this improvement of the rate of
spectral convergence is more than compensated for by the divergence in the constant term.
\end{remark}

The integrals 
in \eqref{eq:boundQ} and \eqref{eq:boundF} can, in some cases, be computed analytically,
or bounded from above by a closed form expression.
Consider for instance dissipative models with finite jump intensity. 
These models are characterised by $\sigma^2>0$ and $\cl=0$. Thus the integrals can be expressed in terms of the cumulative normal distribution $\Phi$:
\begin{align}
\int_{\mathbb R} \expfs{-\tau \frac{\sigma^2 \omega^2}{2}}d\omega&=\sqrt{\frac{2 \pi}{\tau \sigma^2}},
\\
\int_{\varsigma}^\infty \expfs{-\tau \frac{\sigma^2 \omega^2}{2}} d \omega &= 
\sqrt{\frac{2 \pi}{\tau \sigma^2}}
\parent{1-\Phi \parent{\varsigma \sqrt{\tau \sigma^2}}}
\end{align}

Now we consider the case of pure-jump processes (ie, $\sigma^2 =0$) that satisfy the condition $\cl > 0$ for some $\lambda \in (0,2)$. 
In this case the integrals are expressible
in terms of the incomplete gamma function $\gamma$.
First, let us define the auxiliary integral:
\begin{align}
I \parent{a,b} 
&\equiv
\expfs{-a} + a^{-\frac{1}{b}} \gamma \parent{\frac{1}{b},a}
\nonumber
\end{align}
for $a,b>0$. Using this, the 
integrals become:
\begin{equation}
\int_{\R} \expfs{- \tau \frac{\absval{\omega}^{2-\lambda}}{4} \cl \mathbf{1}_{\absval{\omega}>1}}
= 2 \parent{1+ I \parent{\frac{\tau \cl}{4},2-\lambda}}
\end{equation}

\begin{equation}
\int_{\varsigma}^\infty \expfs{- \tau \frac{\absval{\omega}^{2-\lambda}}{4} \cl \mathbf{1}_{\absval{\omega}>1}}
= \begin{cases}
I \parent{\frac{\tau \cl}{4},2-\lambda} + 1-\varsigma & \quad \varsigma <1\\
\varsigma I \parent{\frac{\tau \varsigma^{2-\lambda}\cl}{4},2-\lambda} &\quad \varsigma \geq 1
\end{cases}
\end{equation}
An example of a process for which the previous analysis works
is the CGMY model presented in \cite{cgmy_fine,carr2003stochastic}, for the regime $\cgmyY > 0$.

Lastly, when both $\cl$ and $\sigma^2$ are positive, 
the integrals in \eqref{eq:boundQ} and \eqref{eq:boundF} can be bounded by a simpler expression. 
Consider the two following auxiliary bounds for the same integral, in which $\varsigma\geq 1$:
\begin{align}
\int_{\varsigma}^\infty \expfs{-\tau \parent{ \frac{\sigma^2}{2}\omega^2 + \frac{|\omega|^{2 - \lambda}}{4} \cl} } d\omega
& \leq \expfs{-\tau \frac{\sigma^2}{2}\varsigma^2} \int_{\varsigma}^\infty \expfs{-\tau \frac{|\omega|^{2 - \lambda}}{4} \cl} d\omega \\
& = \varsigma \expfs{-\tau \frac{\sigma^2}{2}\varsigma^2}  I \parent{\frac{\tau \varsigma^{2-\lambda}\cl}{4},2-\lambda} \notag
\end{align}

\begin{align}
\int_{\varsigma}^\infty \expfs{-\tau \parent{ \frac{\sigma^2}{2}\omega^2 + \frac{|\omega|^{2 - \lambda}}{4} \cl} } d\omega
& \leq \expfs{-\tau \frac{\varsigma^{2 - \lambda}}{4} \cl} \int_{\varsigma}^\infty \expfs{-\tau \frac{\sigma^2}{2}\omega^2} d\omega \\
& =\sqrt{\frac{2 \pi}{\tau \sigma^2}} \expfs{-\tau \frac{\varsigma^{2 - \lambda}}{4} \cl} \parent{1-\Phi(\varsigma\sqrt{\tau \sigma^2})}\notag
\end{align}
We have that $b(\varsigma)$, defined as the minimum of the right hand sides of the two previous equations,
\begin{align*}
b(\varsigma)=\min & \left\{ \varsigma \expfs{-\tau \frac{\sigma^2}{2}\varsigma^2}  I \parent{\frac{\tau \varsigma^{2-\lambda}\cl}{4},2-\lambda}, \sqrt{\frac{2 \pi}{\tau \sigma^2}} \expfs{-\tau \frac{\varsigma^{2 - \lambda}}{4} \cl} \parent{1-\Phi\parent{\varsigma\sqrt{\tau \sigma^2}}} \right\}
\end{align*}
is a bound for the integral. Bearing this in mind we have
\begin{align}
\int_{\R} \expfs{-\tau \parent{ \frac{\sigma^2}{2}\omega^2 + \frac{|\omega|^{2 - \lambda}}{4} \cl\indicator_{\absval{\omega}>1}} } d\omega
& \leq 2 \Phi\parent{\sqrt{\tau\sigma^2}}-1 + 2 b(1)
\end{align}
and
\begin{align}
\int_{\varsigma}^{\infty} \expfs{-\tau \parent{ \frac{\sigma^2}{2}\omega^2 + \frac{|\omega|^{2 - \lambda}}{4} \cl\indicator_{\absval{\omega}>1}} } d\omega
\leq & b\parent{\varsigma}
\end{align}
provided that $\varsigma\geq 1$.

\section{Computation and minimization of the bound}
\label{section:Numerics}

In this section, we present numerical examples on the
bound presented in the previous section using practical
models known from the literature.
We gauge the tightness of the bound compared to the
true error using both dissipative and pure-jump processes.
We also demonstrate the feasibility of using the expression of the bound
as a tool for choosing numerical parameters
for the Fourier inversion.

\subsection{Call option in variance gamma model}

The variance gamma model provides a test case
to evaluate the bound in the pure-jump setting. We note
that of the two numerical examples presented, it is the less regular one
in the sense that $\sigma^2=0$ and $C \parent{\lambda} = 0$ for $0 < \lambda < 2$,
indicating that Theorem \ref{teo:boundM} in particular is not applicable. 
 
The \levy measure of the VG model is given by:
\begin{align*} \label{eq:VG}
\nu_{\mathrm{VG}} \parent{\d y} 
= \d y \parent{
\mathbf{1}_{y>0} \frac{\cgmyC \expfs{-\cgmyM y}}{y}
- \mathbf{1}_{y<0} \frac{\cgmyC \expfs{\cgmyG y}}{y}
}
\end{align*}
and the corresponding characteristic function is given by
eq. (7) of \cite{madan1998variance}:
\begin{align*}
\charFun \tau \omega &= 
\parent{1 - i \theta \chi \omega + \frac{\sigma^2 \chi}{2}}^{-\frac{\tau}{\chi}}
\\
\cgmyC &= \chi^{-1}
\\
\cgmyG &= \parent{\sqrt{\frac{\theta^2 \chi^2}{4}+\frac{\sigma^2 \nu}{2}} -\frac{\theta \chi}{2}}^{-1}
\\
\cgmyM &= \parent{\sqrt{\frac{\theta^2 \chi^2}{4}+\frac{\sigma^2 \nu}{2}} +\frac{\theta \chi}{2}}^{-1}
\end{align*}

By Proposition \ref{prop:analy1} 
we get that 
\begin{align} 
a < \mathrm{min} \sset{\cgmyG-\damp,\cgmyM+\damp}
\end{align}
which, combined with the requirement that
$\ga \in L^1 \parent{\R}$ (cf, Remark \ref{remark:callPut}), implies:
\begin{align}
a &< \mathrm{min} \sset{\cgmyM-\damp,\cgmyG+\damp,\damp-1}
\\
a &< \mathrm{min} \sset{\cgmyM-\damp,\cgmyG+\damp,-\damp}
\nonumber
\end{align}
for calls and puts, respectively. We note that evaluation
of the integral in \eqref{eq:fadefinition} is possible also for $\damp \in (0,1)$
and for $\damp < 0$. In fact, there is a correspondence between shifts in the
integration countour and put-call parity. Integrals with $\damp<0$ give rise to put
option prices instead of calls. For an extended discussion of this, we refer to
\cite{lee2004option} or \cite{boyarchenko2011new}, in which 
conformal deformation of the
integration contour is exploited in order to achieve improved numerical accuracy.

In \cite{lee2004option} and in our calculations the parameters equal
$\cgmyM = 39.7840$, $\cgmyG = 20.2648$ and $\cgmyC=5.9311$.

\begin{figure}
\includegraphics[scale=1.0]{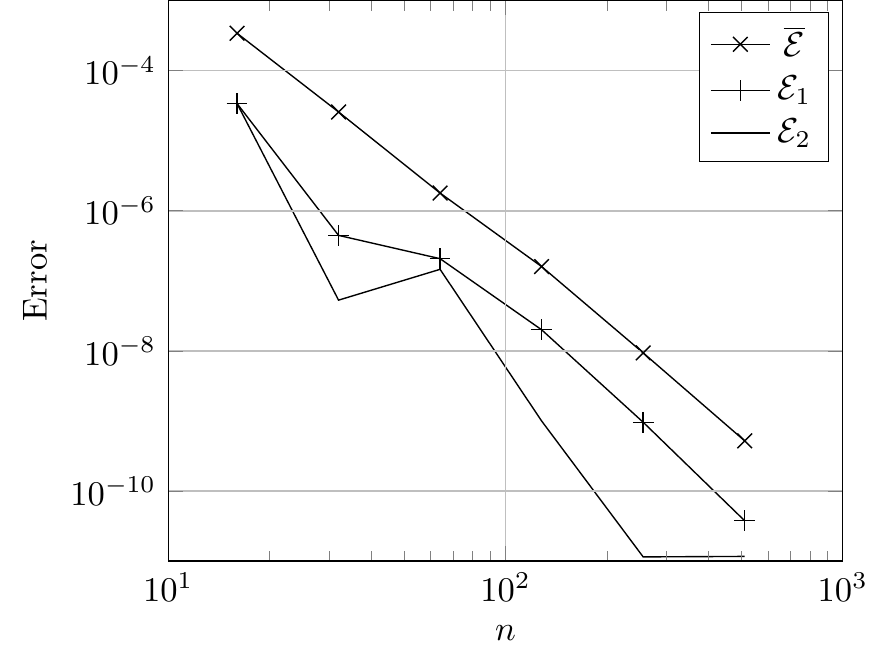}
\caption{\label{fig:vgerrvsbound}
The true error and the error bound for evaluating
at the money options for the VG model test case.
}
\end{figure}

Table \ref{tb:usvslee} presents the specific parameters and compares
the bound for the VG model with the results obtained by \cite{lee2004option}.
Based on the table, we note that for the VG model
presented in \cite{madan1998variance} we can achieve comparable or better error bounds when compared to the study
by Lee. 

To evaluate the bound, we perform the integration of 
\eqref{eq:defM} and \eqref{eq:truncationErrorIntegral}
by relying on the Clenshaw-Curtis quadrature method
provided in the SciPy package.
To supplement Table \ref{tb:usvslee} for a wide range
of $n$,
we present the magnitude of the bound compared to the true error in Figure
\ref{fig:vgerrvsbound}.

In Figure \ref{fig:vgerrvsbound}, we see that the choice of numerical parameters
for the Fourier inversion has a strong influence on the error of the numerical method.
One does not in general have access to the true solution. 
Thus, the parameters need to be optimised 
with respect to the bound.
Recall that $\Error = \Error(\alpha,\Delta \omega,a, n)$ and 
$\est = \est(\alpha,\Delta \omega, n)$ denote the true and estimated
errors, respectively. Keeping the number of quadrature points $n$ fixed,
we let $\parent{\alpha_1,\Delta \omega_1,a_1}$ and $\parent{\alpha_2,\Delta \omega_2}$
denote the minimisers of the estimated and true errors, respectively
\begin{align}
\parent{\alpha_1,\Delta \omega_1,a_1}  &= \mathrm{arg ~ inf}~ \est 
\label{eq:pdef1}
\\
\parent{\alpha_2,\Delta \omega_2}  &= \mathrm{arg ~ inf} ~ \Error 
\label{eq:pdef2}
\end{align}
We further let $\mathcal{E}_1$ and $\mathcal{E}_2$ 
denote the true error as a function of the parameters
minimising the estimated and the true error, respectively
\begin{align}
\mathcal E_1 &= \Error \parent{\damp_1,\step_1} 
\label{eq:pdef3}
\\
\mathcal E_2 &= \Error \parent{\damp_2,\step_2}  
\label{eq:pdef4}
\end{align}
In Figure \ref{fig:vgerrvsbound} we see that the true error increases by approximately an 
order of magnitude when optimising to the bound instead of to the true error,
translating into a two-fold difference
in the number of quadrature points needed for a given tolerance. 
The difference between ${\Error}_1$ 
and the bound is approximately another order of magnitude and necessitates
another two-fold number of quadrature points compared to the theoretical minimum.

\begin{figure}
\subfigure{
\includegraphics[scale=1.0]{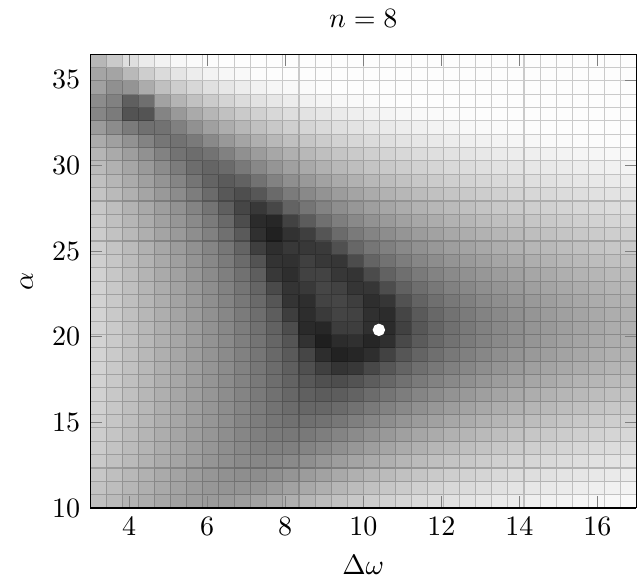}
}
\subfigure{
\includegraphics[scale=1.0]{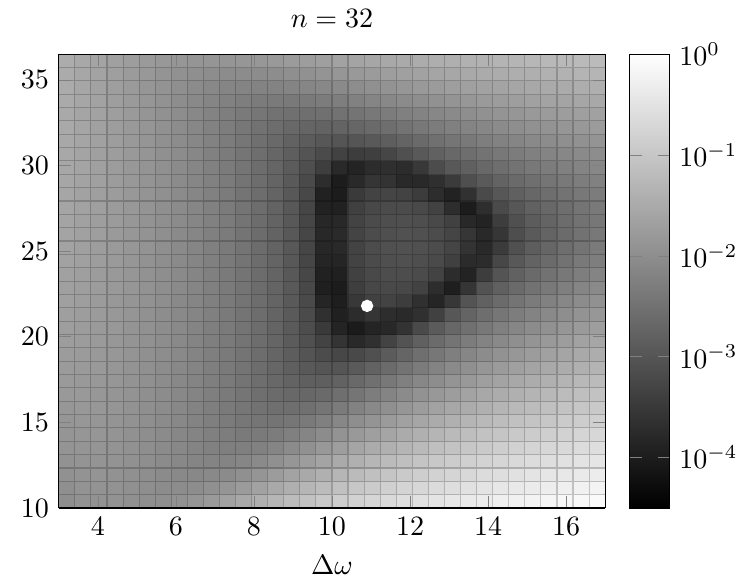}
}
\caption{\label{fig:parameterPlot}
The true error $\Error$
for the two VG test cases presented in Table \ref{tb:usvslee}
and the bound-minimising configurations (white circle) $\parent{\alpha_2,\Delta \omega_2} $
for the examples.
}
\end{figure}

In Figure \ref{fig:parameterPlot}, we present the true error
\footnote{The reference value to compute the true error was
obtained by the numerical methods with $n$ and $\step$ such
that the level of accuracy is of the order $10^{-10}$.} for the Fourier method
for the two test cases in Table \ref{tb:usvslee}. 
 We note that while minimising error
bounds will produce sub-optimal results, the numerical parameters that minimise
the bound are a good approximation of the true optimal parameters. 
This, of course,
is a consequence of the error bound having qualitatively similar behaviour as
the true error, especially as one gets further away from the true optimal parameters.

\begin{remark}
In practice, the Hardy norm in coefficient $M$ reduces to evaluating an
$L^1$ norm along the two boundaries of the strip of width $2a$.
We find that, for practical purposes, the performance of the
Clenshaw-Curtis quadrature of the
QUADPACK library provided by SciPy library is more than adequate,
enabling the evaluation of the bound in a fraction of a second.

For example, the evaluations of the bounds in Table
\ref{tb:usvslee} take only around $0.3$ seconds to evaluate
on Mid 2014 Macbook Pro equipped with a
2.6 GHz Intel Core i5 processor, this without attempting to optimize
or parallelize the implementation and while checking for input sanity
factors such as the evaluation of the characteristic function in a domain
that is a subset in the permitted strip.

We believe that through optimizing routines, skipping sanity checks for inputs
and using a lower-level computation routines this can be optimized even further,
guaranteeing a fast performance even when numerous evaluations are needed.

\end{remark}

\begin{table}
\begin{center}\begin{tabular}{ccccccc}
\hline
 & $K$ & $80$ & $90$ & $100$ & $110$ & $120$\\
\hline
\multirow{2}{*}{$12 \tau=1$} & $\damp$ & $ -16.9 $ & $ -13.8 $ & $ 21.6 $ & $ 29.10 $ & $ 36.3 $ \\ 
 & $a$ & $ 3.33 $ & $ 6.45 $ & $ 18.1 $ & $ 9.77 $ & $ 3.52 $ \\ 
\multirow{2}{*}{$N =32$} & $\cutoff$ & $ 229 $ & $ 229 $ & $ 363 $ & $ 363 $ & $ 424 $ \\ 
 & $\overline{\mathcal E}$ & $ 3.35 \times 10^{-4} $ & $ 0.00334 $ & $ 0.00562 $ & $ 3.97 \times 10^{-4} $ & $ 7.33 \times 10^{-6} $ \\ 
 & $\overline{\mathcal E}^*$ & $ 6 \times 10^{-4} $ & $ 0.0032 $ & $ 0.0058 $ & $ 6 \times 10^{-4} $ & $ 1 \times 10^{-4} $ \\ 
\hline
\multirow{2}{*}{$12 \tau=4$} & $\damp$ & $ -13.8 $ & $ -13.8 $ & $ 22.1 $ & $ 23.7 $ & $ 29.10 $ \\ 
 & $a$ & $ 6.11 $ & $ 6.11 $ & $ 17.9 $ & $ 15.2 $ & $ 8.75 $ \\ 
\multirow{2}{*}{$N =8$} & $\cutoff$ & $ 62.4 $ & $ 42.4 $ & $ 84.9 $ & $ 126 $ & $ 126 $ \\ 
 & $\overline{\mathcal E}$ & $ 3.99 \times 10^{-4} $ & $ 0.00312 $ & $ 0.00398 $ & $ 3.57 \times 10^{-4} $ & $ 1.33 \times 10^{-5} $ \\ 
& $\overline{\mathcal E}^*$ & $ 1.3 \times 10^{-3} $ & $ 0.0057 $ & $ 0.0055 $ & $ 9 \times 10^{-4} $ & $ 1 \times 10^{-4} $ \\

\hline
\end{tabular}\end{center}
\caption{\label{tb:usvslee}The error bound for European call/put options in the VG model for select examples.
Refererence result $\overline{\mathcal E}^*$ from \cite{lee2004option}}
\end{table}

\begin{remark}
Like many other authors, we note the exceptional \emph{guaranteed} accuracy of the FT-method,
with only dozens of quadrature points. This is partially a result of the regularity of the European option price.
Numerous Fourier-based methods have been developed for pricing path-dependent options.
One might, for the sake of generality of implementation, 
be tempted to use these methods for European options as well, correcting for the lack of early exercise
opportunities. This certainly can be done, but due to the
weakened regularity, the required numbers of quadrature points are easily in the thousands, even when no rigorous
bound for the error is required.

We raise one point of comparison,
the the European option pricing example in Table 2 of \cite{jackson2008fourier},
which indicates a number of quadrature points for pricing the option in the range of thousands.
With the method introduced, to guarantee $\est \approx 10^{-3}$,
even with no optimisation, $n=64$ turns out to be sufficient.
\end{remark}

\subsection{Call options under Kou dynamics}

\begin{table}
\begin{center}\begin{tabular}{ccccccc}
\hline
 & $K$ & $80$ & $90$ & $100$ & $110$ & $120$\\
\hline
 & $\overline{\mathcal E}$ & $ 2.67 \times 10^{-4} $ & $  3.49 \times 10^{-4} $ & $  4.43 \times 10^{-4} $ & $ 5.52 \times 10^{-4} $ & $ 6.77 \times 10^{-4} $ \\ 
  & $\damp$ & $ -1.57 $ & $  -1.57 $ & $  -1.57$ & $ -1.57 $ & $ -1.57$ \\ 
  & $\cutoff$ & $ 22.9 $ & $  22.8 $ & $  22.6$ & $ 22.5 $ & $22.4$ \\ 
\hline
  &$\overline{\mathcal E}^*$   & $ 0.34 $ & $  0.26 $ & $  0.21$ & $ 0.17 $ & $0.13$ \\ 
 & $\overline{\mathcal E}^{\dagger}$  & $ 6.87 \times 10^{-4} $ & $  1.90 \times 10^{-3} $ & $  2.82 \times 10^{-3} $ & $ 2.72 \times 10^{-3} $ & $ 2.29 \times 10^{-3} $ \\ 
 \hline
\end{tabular}
\end{center}
\caption{\label{tb:usvsleekou}
Numerical performance of the bound for the Kou model, with the test case in
\cite{toivanen2007numerical} (see also \cite{d2005robust}) with the number
of quadrature points set to $n=32$.
The point of comparison $\overline{\mathcal E}^*$ refers to the corresponding bound computed
with the method described in Chapters $6.1$ to $6.4$ of \cite{lee2004option}.
In the $\overline{\mathcal E}^{\dagger}$, the cutoff error has been evaluated using a 
computationally more intensive Clenshaw-Curtis quadrature instead of an asymptotic argument
with an exponentially decaying upper bound for the option price.
}
\end{table}

For contrast with the pure-jump process presented above, we also test the performance
of the bound for Kou model and present relevant results in \ref{tb:usvsleekou}. This model differs from the first example not only by being dissipative
but also in regularity, in the sense that the maximal width of the domain $A_a$ is, for the 
case at hand, considerably narrower. The \levy measure in the Kou model is given by
\begin{align*}
\nu_{\mathrm{Kou}} \parent{\d y} = \lambda \parent{ p \expfs{-\eta_1 y} \mathbf{1}_{y>0} + q \expfs{\eta_2 y} \mathbf{1}_{y<0} }
\end{align*}
with $p+q=1$.
For the characterisation given in \cite{toivanen2007numerical} the
values are set as
\begin{align*}
\lambda = 0.1, ~~ r = 0.05 ~~ \tau = 0.25 ~~ S_0 = 100 \\
p = 0.3445, ~~ \eta_1 = 3.0465 ~~ \eta_2 = 3.0775
\end{align*}
from the expression of the characteristic exponent (see  \cite{kou2004option})
\begin{align*}
\charExp{z} = z \parent{r -\frac{\sigma^2}{2} - \lambda \zeta} + \frac{z^2 \sigma^2}{2} 
+ \lambda \parent{\frac{p \eta_1}{\eta_1-z} + \frac{q \eta_2}{\eta_2 +z} -1}
\end{align*}
it is straightforward to see
\begin{align*}
A_a \subset \sset{z \in \C : \mathrm{Im} z \in ( -3.0465, 3.0775)}
\end{align*}
This range is considerably narrower than that considered earlier.
In the case of transforming the option prices in strike space, the
relevant expressions for option prices as well as the
error bounds contain a factor exponential in $k$. The practical implication
of this is that for deep out of the money calls, it is often beneficial to exploit
the put-call-parity and to compute deep in the money calls. However, in the case
at hand, the strip width does not permit such luxury. As a consequence,
the parameters that minimize the bound are near-identical through a wide range of moneyness,
suggesting use of FFT algorithm to evaluate the option prices at once for a range of strikes.

\subsection{Binary option in the Merton model}
For the particular case of Merton model, the \levy measure is given by
\begin{align*}
\nu_{\mathrm{Merton}} \parent{\d y}
=
\frac{\lambda}{\sqrt{2 \pi \sigma^2}} \expf{-\frac{\parent{y-r_j}^2}{2\sigma_j^2}}
\end{align*}
and the characteristic exponent correspondingly by
\begin{align*}
\charExp{z}_{\mathrm{Merton}}
=
z \parent{r-\frac{\sigma^2 z}{2}}
+ \frac{\sigma^2 z^2}{2}
+ \lambda
\parent{\expfs{z r_j + \frac{\sigma_j^2 z^2}{2}}-1 - z \parent{\expfs{r_j +\frac{\sigma_j^2}{2} }-1} }
\end{align*}
we may employ a fast semi-closed form evaluation
of the relevant integrals instead of resorting to quadrature methods.
We choose the Merton model as an example of bounding the error of the numerical method for such
a model. The parameters are adopted from 
the estimated parameters for S\&P 500 Index from 
\cite{andersen2000jump}:
\begin{align*}
S_0 = 100,
~~~
\lambda = 0.089,
~~
\sigma = 0.1765,
~~~
r = 0.05,
~~
r_j = -0.8898,
~~~
\sigma_j = 0.4505
\end{align*}

In Figure \ref{fig:mertonPlot}, we present the bound and true error for the Merton model
to demonstrate the bound on another dissipative model. The option presented is a binary option
with finite support on $[95,105]$; no damping was needed or used. We note that
like in the case of the pure-jump module presented above, our bound reproduces the
qualitative behaviour of the true error. The configuration resulting from optimising the bound
is a good approximation of the true error. Such behaviour is consistent through the range of
$n$ of the most practical relevance.

\begin{figure}
\subfigure{
\includegraphics[scale=1.0]{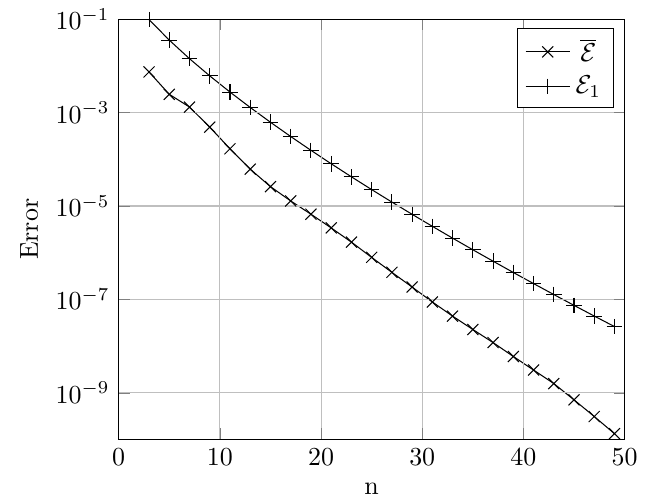}
}
\subfigure{
\includegraphics[scale=1.0]{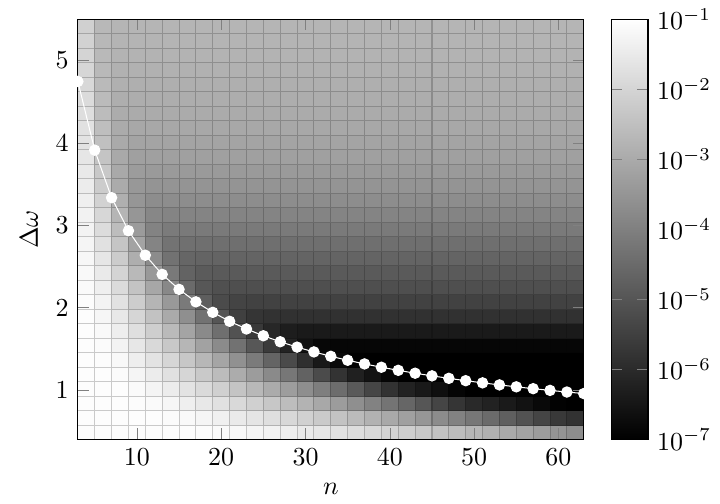}
}
\caption{\label{fig:mertonPlot}
The true error ${\Error}_1$
and the bound  $\est$
for the dissipative Merton model, for a range of
quadrature points $n$, along with the
bound-minimising configurations contrasted to the
true error.
}
\end{figure}
\subsection{Call options}
\label{subsec:call}
In Subsection \ref{ss:bounds} explicit expressions to bound $\Error$ are provided. 
To evaluate these bounds it is necessary to compute $\norminf{\ft{\ga}}{\R}$ and $\norminf{\ft{\ga}}{\strip{a}}$. According to Remark \ref{rem:decuople}, once
we compute these values we could use them for any model, provided that 
they satisfy the conditions considered there.

The payoff of perhaps the most practical relevance is that of
a call option. 
Consider $g$ defined by:
\begin{align*}
 g\parent{x} = \parent{S_0 \expfs{x}-K}^+ = S_0 \parent{\expfs{x}-\expfs{k}}^+
\end{align*}
for which a selection of a damping parameter $\damp>0$ is necessary
to have the damped payoff in $L^1 \parent{\R}$ and to ensure the existence of a
Fourier transformation. In this case we have
\begin{align}
\hat{\ga}
\parent{\omega} &= S_0 \int_\R \expf{ \parent{1-\damp + i \omega}x} - \expf{k+\parent{i \omega- \damp}x}
\d x
\\
   &=
   \frac{ S_0 \expf{\parent{1- \damp + i \omega} k } }{\parent{1 + i \omega- \damp} \parent{i \omega - \damp}}
\end{align}
and
\begin{align}
 \absval{\ft{\ga}(\omega) }^2
 =
 \frac{S_0^2 \expfs{2 \parent{1-\alpha}k}}{\parent{\alpha^2 + \omega^2 } \parent{\parent{1-\alpha}^2 + \omega^2}}
\end{align}
It is easy to see that the previous expression decreases as $|\omega|$ increase. 
This yields
\begin{align}
\norminf{\ft{\ga}}{\R}
=\absval{\ft{\ga}(0) } =
\frac{S_0 \expfs{\parent{1-\alpha}k}}{\alpha^2-\alpha}
\end{align} 
and 
\begin{align}
\norminf{\ft{\ga}}{[\varsigma,\infty)}
=\absval{\ft{\ga}(\varsigma) }
\end{align} 
The maximisation of $\absval{\ft \ga}$ in the strip $\strip{a}$ of the complex plane is more subtle.
Denoting $\ft \ga(\eta,\rho)=\ft \ga(\eta + i \rho)$, we look for critical points that
satisfy 
$\partial_{\eta} \absval{\ft \ga}=0$. 
This gives
\begin{align}
4 \eta^3 + 2 \eta \parent{4 \rho \alpha + 2 \alpha^2 - 2 \rho -2 \alpha + \rho^2 +1}=0.
\end{align}
For $\rho$ fixed, $\absval{\ft \ga}$ has a vanishing derivative with respect to $\eta$
at a maximum of three points. Of the three roots of the derivative, only the one characterised
by $\eta=0$ is a local maximum, giving us that for call options
\begin{align}
\norminf{\ft{\ga}}{\strip{a}} = \sup_{\rho \in [-a,a]} \absval{\ft \ga(0,\rho)} 
\label{eq:complexsup}
\end{align}
Now, observe that $\absval{\ft \ga(0,\rho)}$ is a differentiable real function of $\rho$, whose derivative is given by the following polynomial of second degree:
\begin{align}
 p \parent{\rho} \equiv k \parent{\rho +\alpha - 2 \rho \alpha - \alpha^2 - \rho^2}
 - 2 \alpha - 2 \rho +1
\end{align}
We conclude that
\begin{align}
\norminf{\ft{\ga}}{\strip{a}} = \max_{\rho \in B} \left\{ \absval{\ft \ga(0,\rho)}\right\} 
\label{eq:complexsup2}
\end{align}
where $B$ is the set of no more than four elements consisting of $a$; $-a$; and the real roots of $p$ that fall in $(-a,a)$.

\begin{remark}
So far, we have assumed the number of quadrature points $n$ to be constant.
In real life applications, however, this is often not the case. Typically the user might want to choose a minimal
$n$ that is sufficient to guarantee error that lies within a pre-defined error tolerance.

In such a case, we propose the following, very simplistic scheme for optimising numerical parameters
and choosing the appropriate $n$ to satisfy error smaller than $\epsilon$:
\begin{enumerate}
\item Select $n=n_0$ and optimize to find the relevant configuration
\item See if $\EQ + \EF <\epsilon$, if not, increase $n$ by choosing it from a pre-determined
increasing sequence $n=n_j$ and repeat procedure.
\end{enumerate}
Especially in using FFT algorithms to evaluate the Fourier transforms,
we propose $n_j = 2^j n_0$. We further note that typically the optimal configuration
for the optimizing configuration for $n_{j+1}$ quadrature points does not differ too
dramatically from the configuration that optimizes bounds for $n_j$.
\end{remark}

\section{Conclusion}
\label{section:Conclusion}

We have presented a decomposition of the error committed in
the numerical evaluation of the inverse Fourier transform needed
in asset pricing for exponential \levy models into truncation and
quadrature errors. For a wide class of exponential \levy models, we
have presented an $L^\infty$-bound for the error.

The error bound differs from the earlier work presented in \cite{lee2004option}
in the sense that it does not rely on the asymptotic behaviour of the option
payoff at extreme strikes or option prices, allowing pricing a wide variety
of non-standard payoff functions such as the ones in \cite{suh2008class}.
The bound, however, does not take into account path-dependent options. We
argue that the error for the methods that allow evaluating american, bermudan,
or knockoff options are considerably more cumbersome and produce significantly
larger errors so that in implementations where performance is important, such as
calibration, using American option pricing methods for European options is not justified.

The bound also provides a general framework in which
the truncation error is evaluated using a quadrature method that remains invariant
regardless of the asymptotic behaviour of the option price function.
The structure of the bound allows for
a modular implementation that decomposes the error components arising from
the dynamics of the system and the payoff into a product form for a large
class of models, including all dissipative models. On select examples, we
also demonstrate the performance that is comparable or superior to the
relevant points of comparison.

We have focused on the minimization of the bound as a proxy for minimizing numerical error.
Doing this, one obtains, for a given parametrization of a model, a rigorous $L^\infty$
bound for the error committed in solving the European option price. We have shown that the 
bound reproduces the qualitative behaviour of the actual error. This supports the argument for
selecting numerical parameters in a way that minimizes the bound, giving evidence that
this selection will, besides guaranteeing numerical precision, be close to the actual minimizing
configuration that is not often achievable at an acceptable computational cost.

The bound can be used in the primitive setting of establishing a strict error
bound for the numerical estimation of option prices for a given set of physical and
numerical parameters or as a part of a numerical
scheme, whereby the end user wishes to estimate an option price either on
a single point or in a domain up to a predetermined error tolerance.

In the future, the error bounds presented can be used in efforts requiring
multiple evaluations of Fourier transformations. Examples of such applications
include multi-dimensional Fourier transformations, possibly in sparse tensor grids,
as well as time-stepping algorithms for American and Bermudan options.
Such applications are sensitive towards the error bound
being used, as any numerical scheme will be required to run multiple times, either
in high dimension or for multiple time steps (or both).

\section{Acknowledgements}

H\"app\"ol\"a, Crocce are members and Tempone is the director KAUST Strategic Research Initiative
for Uncertainty quantification.

We wish to thank an anonymous referee for their critique that significantly improved the quality of the manuscript.

\section{Declarations of Interest}

The authors report no conflicts of interest. The authors alone are responsible for the content and writing of the paper.

\appendix

\section{Truncation bound in Lee's scheme}

For the strike-space transformed option price we have
\begin{align*}
\tilde \fa \parent{\omega}
=
\frac{\charFun 1 {\omega - i \parent{\damp + 1}}}{\parent{i \omega + \damp} \parent{i \omega + \damp + 1}}
\end{align*}
In the particular case of the Kou model, we have 
\begin{align*}
\charFun 1 \omega
=
\expfs{\tau
\parent{
i \omega \parent{r - \frac{\sigma^2}{2} - \lambda \zeta}
- \frac{\omega^2 \sigma^2}{2}
+ \lambda
\parent{
\frac{p \eta_1 }{\eta_1 - i \omega}
+\frac{q \eta_2}{\eta_2 + i \omega}
-1
}
}
}
\end{align*}
Thus,
\begin{align*}
\absval{ \charFun 1\omega}
=
\expfs{
\tau
\parent{ 
\parent{\damp+1}
\parent{r - \frac{\sigma^2}{2} - \lambda \zeta}
+
\frac{\parent{\damp+1}^2 - \omega^2}{2}
}
}
\expfs{ \tau \lambda \mathrm{Re} \parent{
\frac{p \eta_1}{\eta_1-\damp-1 - i \omega}
\frac{q \eta_2}{\eta_2 +\damp +1 + i \omega}
-1
 }}
\end{align*}
We bound the first jump term resulting from upward jumps as
\begin{align*}
\frac{p \eta_1}{ \eta_1 - \damp -1 - i \omega}
= \frac{p \eta_1 \parent{\eta_1 - \damp -1 +i \omega}}{\parent{\eta_1 - \damp -1}^2 + \omega^2}
\end{align*}
from which it follows
\begin{align*}
\mathrm{Re}
\parent{
\frac{p \eta_1}{ \eta_1 - \damp -1 - i \omega}
}
\leq
\frac{p \eta_1 \parent{\eta_1 - \damp -1 }}{\parent{\eta_1 - \damp -1}^2 }
\end{align*}
similar inequality holds for the term resulting from downward jumps giving us
that
\begin{align*}
\absval{ \charFun 1 \omega}
\leq 
\expfs{
\tau
\parent{ 
\parent{\damp+1}
\parent{r - \frac{\sigma^2}{2} - \lambda \zeta}
+
\frac{ \sigma^2 \parent{\damp+1}^2 - \sigma^2 \omega^2}{2}
}
}
\expfs{
\tau \lambda
\parent{
\frac{p \eta_1 \parent{\eta_1 - \damp -1}}{ \parent{\eta_1 -\damp -1}^2}
+
\frac{q \eta_2 \parent{\eta_2 + \damp +1}}{\parent{\eta_2+ \damp +1}^2}
+1
}
}
\end{align*}
The term $\expf{- \frac{\omega^2}{2}}$ can be bounded from above
by an exponential term $\expf{-\gamma \omega}$:
\begin{align}
\expf{- \frac{\sigma^2 \omega^2}{2}}
\label{eq:tailBound}
& \leq
\expf{-\gamma \omega+ K}
\\
- \sigma^2 \omega^2 + 2 \gamma \omega - 2 K & \leq 0
\nonumber
\end{align}
the condition for the quadratic to have no solutions
is given by
\begin{align*}
 \gamma^2 - 2 K  \sigma^2 \leq 0,
\end{align*}
setting $K =\frac{1}{2}$, $\gamma = \sigma$
giving us the bound in eq. \eqref{eq:tailBound}.

This puts us in place to apply theorem 6.1 of \cite{lee2004option}, with
\begin{align*}
\Phi \parent{u} &= 
\frac{
\expf{
\tau
\parent{ 
\parent{\damp+1}
\parent{r - \frac{\sigma^2}{2} - \lambda \zeta}
+
\frac{ \sigma^2 \parent{\damp+1}^2 }{2}
}
}
\expfs{
\tau \lambda
\parent{
\frac{p \eta_1 \parent{\eta_1 - \damp -1}}{ \parent{\eta_1 -\damp -1}^2}
+
\frac{q \eta_2 \parent{\eta_2 + \damp +1}}{\parent{\eta_2+ \damp +1}^2}
+1
}
}}
{\damp^2 \parent{\damp+1}^2}
\\
\gamma & = \sigma
\end{align*}

\end{document}